\documentclass[11pt,a4paper]{article}

\usepackage[a4paper,margin=1in]{geometry}
\usepackage{amsmath, amssymb, amsthm, bm}
\usepackage{graphicx}
\usepackage{booktabs, multirow, tabularx, threeparttable}
\usepackage{natbib}
\usepackage{setspace}
\usepackage{hyperref}
\usepackage{enumitem}
\usepackage{xcolor}
\usepackage{float}
\usepackage{authblk}
\usepackage{mathtools}
\usepackage{rotating}
\usepackage{tikz}
\usepackage[edges]{forest}
\usetikzlibrary{arrows.meta,calc,positioning,fit}
\usepackage{listings}
\usepackage{courier}  

\usepackage{amsthm}
\newtheorem{theorem}{Theorem}[section]
\newtheorem{proposition}[theorem]{Proposition}

\newtheorem{assumption}[theorem]{Assumption}

\theoremstyle{definition}
\newtheorem{definition}[theorem]{Definition}

\theoremstyle{remark}
\newtheorem{remark}{Remark}[section]

\title{Rethinking Individual Risk and Aggregation in Survival Analysis: A Latent Mechanism Framework}
\author{
  Xijia Liu\thanks{
    Department of Statistics, Ume\aa\ University, Sweden.
    \textbf{Email}: \texttt{xijia.liu@umu.se}
  }
}
\date{}

\providecommand{\keywords}[1]{\textbf{\textit{Keywords---}} #1}

\begin{document}
\maketitle
\begin{abstract}
    Survival analysis provides a well-established framework for modeling time-to-event data, with hazard and survival functions formally defined as population-level quantities.
    In applied work, however, these quantities are often interpreted as representing individual-level risk, despite the absence of a clear generative account linking individual risk mechanisms to observed survival data.
    This paper develops a latent hazard framework that makes this relationship explicit by modeling event times as arising from unobserved, individual-specific hazard mechanisms and viewing population-level survival quantities as aggregates over heterogeneous mechanisms.
    Within this framework, we show that individual hazard trajectories are not identifiable from survival data under partial information.
    More generally, the conditional distribution of latent hazard mechanisms given covariates is structurally non-identifiable, even when population-level survival functions are fully known.
    This non-identifiability arises from the aggregation inherent in survival data and persists independently of model flexibility or estimation strategy.
    Finally, we show that classical survival models can be systematically reinterpreted according to how they handle this unresolved conditional mechanism distribution.
    This paper provides a unified framework for understanding heterogeneity, identifiability, and interpretation in survival analysis, and clarifies how population-level survival models should be interpreted when individual risk mechanisms are only partially observed, thereby establishing explicit information constraints for principled modeling and inference.
\end{abstract}

\section{Introduction}
Survival analysis has achieved remarkable empirical success over the past five decades and has become a central tool in biomedical research, epidemiology, and reliability studies.
Its fundamental quantities, the hazard function and the survival function, are defined with mathematical rigor and stability, and have supported extensive methodological development and applied practice.
However, alongside this success, a persistent interpretational tension has remained present but has never been systematically clarified.

Within the classical theoretical framework, hazard and survival functions are defined as population-level statistical objects, characterizing failure rates and survival probabilities within hypothetical cohorts sharing the same observed covariates.
In applied settings, however, these same quantities are often naturally endowed with individual-level interpretations and are used to describe individual risk states, latent biological processes, or uncertainty about future outcomes.
This issue becomes particularly salient as the focus of survival analysis has increasingly shifted from inference-oriented to prediction-oriented objectives, and has become difficult to avoid in contexts such as individualized prediction, risk stratification, and decision support.
More importantly, this tension is not confined to any particular class of models or methods, but reflects a structural information constraint: survival data inherently aggregate over heterogeneous, unobserved hazard mechanisms, and this aggregation induces a many-to-one mapping from individual-level risk processes to observable population-level quantities.
Clarifying this issue is therefore of foundational importance for understanding the interpretational boundaries of survival models, the meaning of prediction, the feasibility of heterogeneity characterization, and for providing principled guidance for the design of future predictive models.

\subsection{Motivating Challenges in Survival Analysis}
The source of the interpretational tension outlined above is not immediately apparent at the level of formal definitions.
In classical survival theory, hazard and survival functions are rigorously defined as population-level quantities.
In particular, the standard hazard
\[
h(t | X) = \lim_{\Delta t \to 0}
\frac{\Pr(t \le T < t+\Delta t | T \ge t, X)}{\Delta t},
\]
characterizes the instantaneous failure rate among individuals who have survived up to time \(t\) within a hypothetical cohort sharing the same observed covariates.
By construction, this quantity is conditional on membership in the risk set and is defined only at the level of such populations, rather than as an intrinsic property of a single subject.
Nevertheless, in applied contexts, both classical and modern survival models are often discussed, reported, or informally interpreted in ways that suggest an individual-level notion of risk, even though \(h(t| X)\) and \(S(t| X)\) are formally defined as population-level quantities.

This interpretational shift is closely tied to unobserved heterogeneity.
Individuals who share the same observed covariates may differ substantially in their underlying risk processes.
Therefore, population-level hazards represent averages over heterogeneous individuals rather than intrinsic individual properties.
Unless heterogeneity vanishes, a condition rarely met in biomedical settings, population hazards cannot be directly equated with individual risks.
However, this distinction is often blurred in practice, complicating interpretation, prediction, and causal reasoning.

Classical regression-based survival models adopt specific functional forms for population-level quantities.
For example, Cox's proportional hazards model \cite{cox1972regression} and accelerated failure time (AFT) \cite{Buckley1979AFT} formulations describe structured relationships between covariates and hazards or survival times.
These models are highly effective descriptive tools, but they do not explicitly specify the individual-level processes that would give rise to the assumed population structures.
As a result, the connection between model assumptions and underlying individual risk dynamics remains implicit.

Recent machine learning approaches to survival analysis intensify this issue.
By learning flexible mappings from covariates to time-dependent hazards or survival probabilities, such models can further blur the distinction between population-level quantities and individual-level interpretations, even though these quantities remain formally defined at the population level.
Without a generative framework that distinguishes individual risk mechanisms from their population aggregates, the interpretation of model outputs remains ambiguous.
Taken together, these considerations highlight the absence of an explicit mechanism-based foundation that links individual-level risk processes to observed population-level survival patterns.

\subsection{Existing Perspectives on Survival Modeling}
The distinction between individual-level and population-level notions of hazard has long been implicit in the survival analysis literature.
In Cox’s original formulation \citep{cox1972regression}, the hazard function is introduced as a regression object describing cohort-level failure rates, with the baseline hazard treated as an arbitrary population function of time.
The model itself does not attribute ontological status to individual hazards, nor does it attempt to specify individual risk processes.
However, the absence of an explicit generative interpretation has left room for divergent readings in subsequent methodological and applied work.

A number of authors have emphasized interpretational subtleties associated with hazard-based quantities.
Hazard contrasts compare failure rates among individuals who have survived up to a given time rather than individual risks \citep{aalen2008survival, martinussen2013collapsibility, Grambsch2017}, and unobserved heterogeneity can distort hazard-based interpretations even in randomized studies \citep{dumas2025hazard}.
These contributions provide important clarification regarding the limitations of hazard ratios and related summaries, but they do not offer a general framework for linking population hazards to underlying individual risk processes.
Parallel developments have sought to address heterogeneity through extensions such as frailty models \cite{vaupel1979impact, hougaard1995frailty}, accelerated failure time formulations, and latent class survival analysis \cite{lin2002latent, proust2014joint}.
Each of these approaches introduces additional structure to account for variability beyond observed covariates, whether through random effects, discrete subpopulations, or time-scaling mechanisms.
More recently, machine learning based survival models have emphasized predictive performance by learning flexible mappings from covariates to observable risk or survival quantities \citep{Ishwaran2008RSF, Katzman2018, lee2018deephit}; see \citet{Wang2021Survey, wiegrebe2024deep} for reviews.
While these methods differ substantially in form and intent, they share a common feature: the underlying relationship between individual-level risk mechanisms and population-level survival quantities remains implicit, particularly in predictive settings where individual risk estimates are difficult to validate or calibrate \citep{Steyerberg2010assessing, Austin2020graphical}.

What is currently missing is a unifying perspective that makes this relationship explicit.
In particular, there is a need for a framework that (i) defines individual risk processes as latent mechanisms, (ii) explains how population-level hazards arise as aggregates of these mechanisms, and (iii) clarifies the information-theoretic limitations that constrain what can be inferred from survival data.
The present work is motivated by this gap.

\subsection{A Mechanism-Based Framework for Survival Analysis}
To address the interpretational and identifiability challenges outlined above, this paper proposes a mechanism-based analytical framework that explicitly characterizes individual-level risk-generating mechanisms and clarifies how population-level survival quantities arise as aggregated consequences of these latent mechanisms, thereby laying a foundation for a unified understanding of classical survival models.

Survival analysis is fundamentally concerned with the relationship between individual characteristics and time-to-event outcomes.
Implicit in many modeling approaches, and in their common interpretations, is the notion that each individual is associated with a specific, though unobserved, hazard trajectory over time.
We make this idea explicit by introducing the notion of an \emph{individual hazard mechanism}, denoted by $\Theta$, which characterizes the stochastic process governing an individual’s event time.
Given a realization of $\Theta$, an individual hazard function and survival function are well defined, but the mechanism itself is not directly observable.
Crucially, survival data provide extremely limited information about individual mechanisms.
Each individual typically contributes at most a single event-time observation, possibly subject to censoring.
As a consequence, even in the presence of rich covariate information, individual hazard mechanisms are fundamentally not identifiable from survival data alone, even under idealized sampling and arbitrarily large sample sizes.
This limitation is not a technical artifact of specific modeling choices, but a fundamental consequence of the information structure inherent in survival data.
Observed survival quantities, therefore, admit a natural population-level interpretation.
Conditional on covariates $X=x$, the survival function $S(t | X=x)$ arises as an aggregate over the distribution of individual mechanisms compatible with these covariates.
Rather than representing the survival trajectory of a representative individual, population-level survival reflects a mixture of heterogeneous hazard mechanisms.
This distinction is essential for interpreting model outputs and for understanding the scope of valid inference in survival analysis.
From this perspective, the central object linking individual-level heterogeneity and population-level survival is the conditional mechanism distribution $\mathbb{P}(\Theta | X=x)$.
This distribution encodes how much information observed covariates provide about individual hazard mechanisms.
However, because survival data do not resolve individual mechanisms, this conditional distribution is typically not identifiable beyond coarse, model-dependent summaries.
Any attempt to model survival data must therefore confront the fact that $\mathbb{P}(\Theta|X=x)$ cannot be freely estimated from the data alone.

Survival models can thus be understood as imposing structural assumptions on the conditional mechanism distribution.
Different modeling traditions correspond to different ways of constraining residual heterogeneity: by suppressing it, by restricting it to low-dimensional forms, or by representing it through discrete components.
These assumptions are often implicit, yet they fundamentally shape what aspects of survival behavior can be inferred from the data and how model outputs should be interpreted.

The aim of the present framework is not to recover individual hazard mechanisms but to provide a coherent language for articulating heterogeneity, information limitations, and modeling assumptions in survival analysis.
By making the role of latent mechanisms explicit, the framework clarifies what can and cannot be identified from survival data and sets the stage for a systematic examination of identifiability, which we develop in the following sections.

\subsection{Contributions}
This paper makes a conceptual and structural contribution to the understanding of survival analysis.
Rather than proposing a new estimation procedure or predictive model, our focus is on clarifying the objects of inference, the role of latent heterogeneity, and the information-theoretic limitations that shape survival modeling.
By reframing survival analysis through the lens of individual hazard mechanisms, we aim to provide a coherent perspective that unifies and contextualizes a broad range of existing approaches.

Specifically, the contributions of this work are as follows.
First, we introduce a mechanism-based formulation that explicitly distinguishes between individual-level hazard mechanisms and population-level survival quantities, thereby resolving long-standing ambiguities in interpretation.
Second, within this framework, we characterize the conditional mechanism distribution as a central but generally non-identifiable object in survival analysis, and show that this non-identifiability arises from the intrinsic information structure of survival data rather than from methodological or modeling deficiencies.
Building on these results, we reinterpret several classical survival models as imposing distinct structural constraints on the conditional mechanism distribution, providing a unified explanation for their assumptions, representational scope, and inherent limitations.
Finally, as a unifying consequence, the proposed framework offers a principled language for comparing modeling strategies and for articulating directions for future methodological development in survival analysis.
By making the latent information structure explicit, the proposed framework not only clarifies the limits of individual-level inference, but also delineates a space of modeling choices in which methodological assumptions can be stated, compared, and evaluated in a principled manner.

\subsection*{Structure of the Paper}
The remainder of the paper is organized as follows.
Section~\ref{sec:framework} introduces the latent hazard mechanism framework and formalizes the relationship between individual mechanisms, observable covariates, and survival outcomes.
Section~\ref{sec:identifiability} examines the identifiability implications of this formulation, establishing general limitations imposed by the information structure of survival data.
In Section~\ref{sec:classical}, we revisit several classical survival models, including proportional hazards models, frailty models, accelerated failure time models, and survival clustering approaches, through the proposed framework, highlighting how each handles residual heterogeneity.
Section~\ref{sec:discussion} concludes with a discussion of implications, limitations, and directions for future research.

\section{Latent Hazard Framework}
\label{sec:framework}

A persistent source of ambiguity in survival analysis concerns the relationship between individual-level risk and population-level observable quantities.
While hazard and survival functions are routinely interpreted as characterizing individual risk, they are, in practice, defined and estimated at an aggregated level, conditional on limited information.
To analyze the identifiability and interpretation of such quantities, it is therefore necessary to make explicit the information structure linking individual risk-generating mechanisms, observable covariates, and survival outcomes.

To this end, we introduce a minimal latent hazard formulation that explicitly separates individual-level risk mechanisms from observable population-level quantities.
At the core of this formulation is a latent hazard mechanism $\Theta$, which represents individual-level risk-generating structure and is treated as a primitive object that deterministically induces an individual hazard and survival trajectory.
Observable covariates $X$ are interpreted as partial information about this latent mechanism, giving rise to a conditional distribution over $\Theta$, rather than directly modifying individual risk once the mechanism is fixed.

Within this perspective, classical population-level and group-level survival quantities arise as observable summaries obtained by aggregating over latent heterogeneity in $\Theta$.
In particular, survival functions correspond to mixture averages of individual survival trajectories, while observable hazards emerge as survivor-weighted averages of individual hazards.
The resulting formulation is fully consistent with the standard formalism of survival analysis, but makes explicit the distinction between individual-level mechanisms and group-level observable quantities.
This separation provides the conceptual foundation for the identifiability analysis developed in the subsequent sections.

\subsection{Individual Hazard Mechanism}
\label{sec:individual_hazard_mechanism}

We formalize the notion of an individual hazard mechanism, which serves as a conceptual representation of individual-level heterogeneity in survival outcomes.

\begin{definition}[Individual hazard mechanism]
    Let \((\Omega,\mathcal{F}, \mathbb{P})\) be an underlying probability space. 
    An individual hazard mechanism is a random element
    \[
        \Theta : (\Omega,\mathcal{F}) \to (\mathcal{M},\mathcal{B}),
    \]
    taking values in a measurable space \((\mathcal{M},\mathcal{B})\), together with a measurable mapping
    \[
        \mathcal{H} : \mathcal{M} \to \{ h : [0,\infty) \to [0,\infty) \},
    \]
    such that each realization \(\theta \in \mathcal{M}\) deterministically induces an individual hazard trajectory
    \(h_\theta(t) = \mathcal{H}(\theta)(t),\)
    \(t \ge 0.\)
    Thus, the hazard function of an individual is not itself random beyond the randomness of $\Theta$; all heterogeneity in individual time-to-event behavior is represented, at an abstract level, through the distribution of $\Theta$.
\end{definition}

\begin{remark}[Mechanism space versus function space]
    Although each realization of an individual hazard mechanism ultimately induces a nonnegative hazard function through the mapping $\mathcal H$, we do not define $\Theta$ directly as a random function.  
    Instead, $\Theta$ is treated as an abstract random element taking values in a general mechanism space $\mathcal M$, whose realizations deterministically generate hazard trajectories.  
    This distinction separates the representation of individual risk-generating structure from the functional form of the resulting hazard and allows additional structure to be incorporated at the mechanism level.
\end{remark}

The mechanism $\Theta$ is intentionally defined at an abstract level, and the framework does not require a unique parametrization of individual risk.  
Its purpose is to formalise heterogeneity at the level of hazard-generating structure itself, rather than as stochastic variation added to an otherwise homogeneous event-time distribution.  
Under this interpretation, individuals differ because their underlying mechanisms differ, and these differences induce distinct hazard trajectories.
As a purely illustrative example, one convenient way to instantiate such mechanisms is through a hierarchical parametric representation.
In this formulation,
\[
    \Theta = (\mathcal C, \Phi),
\]
where $\mathcal C$ denotes a hazard class determining the qualitative shape of the hazard function, such as monotone, unimodal, or bathtub-shaped behaviour, and $\Phi$ collects class-specific parameters specifying a particular trajectory within that class.  
The induced hazard function is then given by $h_\Theta(t) = h_{\mathcal C}(t;\Phi)$.  
This representation aligns with common taxonomic descriptions of hazard shapes and provides a concrete example of how a mechanism may encode both coarse structural patterns and within-class variability.
Importantly, the general framework does not rely on this or any other specific parametrisation.  
The admissible hazard space may be chosen according to the scientific context and modelling objectives, ranging from classical parametric families to geometric classes, mixtures, or more abstract process-based representations.  
Accordingly, the latent mechanism $\Theta$ should be viewed as a general probabilistic schema, while concrete models are obtained by specifying an appropriate hazard space for the application at hand.

\paragraph{From hazard mechanisms to individual survival.}
An individual hazard trajectory $h_\Theta(t)$ is interpreted as a rate of risk accumulation over time.
Accordingly, it induces a cumulative hazard
\[
    H_\Theta(t)=\int_0^t h_\Theta(u)\,du,
\]
which summarizes the total amount of risk accrued up to time $t$.
To connect this risk accumulation with event-time outcomes, we introduce an event-time random variable $T:\Omega\to[0,\infty]$ defined on the same probability space.
We assume that, conditional on the realization of the hazard mechanism, the distribution of $T$ is fully determined by the induced hazard trajectory.
Specifically, if $h_\theta$ is locally integrable, the associated conditional survival function is given by
\[
    S_\Theta(t)=\exp\!\left(-H_\Theta(t)\right).
\]
For a fixed realization $\Theta=\theta$, this survival curve coincides with the conditional survival law
\[
    S_\theta(t)=\mathbb P(T>t| \Theta=\theta).
\]
This is a well-defined probability statement at the level of the underlying probability space.
However, it should not be interpreted as a time-evolving probability assessment for a fixed realized individual.
Rather, $S_\theta(t)$ characterizes the conditional distribution of event times across individuals sharing the same hazard-generating mechanism, and serves as a deterministic functional summarizing how risk accumulates along the trajectory specified by $\theta$.

\begin{remark}[Mechanism stability]
    Throughout this paper, the individual hazard mechanism $\Theta$ is assumed to be stable over the observation window.
    Time enters the model through the evaluation of the hazard trajectory $h_\Theta(t)$ and through the accumulation of risk, rather than through structural changes in the mechanism itself.
    Under this stability assumption, any temporal variation in individual risk is expressed through the form of the hazard trajectory $h_\Theta(t)$, rather than through the survival function.
    This assumption is implicit in classical survival models and is adopted here as a baseline for theoretical clarification, rather than as a claim about biological reality.    
\end{remark}

\subsection{Group-level survival and observable hazard}
\label{sec:group-level}

In this subsection, we clarify how classical population-level and group-level survival quantities arise as observable summaries of latent individual risk mechanisms.
The purpose of this section is not to introduce new modeling assumptions, but to make explicit the informational structure through which latent heterogeneity is aggregated into observable survival and hazard functions.

\subsubsection{Observable information and conditional mechanism distributions}

We begin by formalizing the role of observable covariates as carriers of partial information.

\begin{definition}[Observable information]
    Let $(\Omega,\mathcal F,\mathbb P)$ be the underlying probability space introduced in Section~\ref{sec:individual_hazard_mechanism}.
    An observable information is a random element
    \[
    X : (\Omega,\mathcal F)\to(\mathcal X,\mathcal A)
    \]
    defined on the same probability space.
    The observable information $X$ is interpreted as partial information about the latent mechanism $\Theta$.
    Formally, $X$ induces a regular conditional distribution
    \[
    \mathbb P_\Theta(\cdot | X)
    \]
    on $(\mathcal M,\mathcal B)$, representing the residual heterogeneity in hazard mechanisms that remains after conditioning on the available observations.
\end{definition}

This definition is deliberately informational rather than generative.
No functional, causal, or structural relationship between $X$ and $\Theta$ is imposed.
Conditioning on $X$ updates the distribution of latent mechanisms but does not alter the mechanism itself.
It is also worth noting that conditioning on survival up to time $t$ provides additional observable information and induces a time-dependent conditional distribution $\mathbb P(\Theta | T \ge t, X)$.
Consequently, all risk-generating content is attributed to the mechanism, a principle formalized below as mechanism sufficiency.

\paragraph{Mechanism sufficiency.}

Within the proposed formulation, the latent mechanism $\Theta$ is interpreted as encoding the complete individual-level risk-generating structure.
Once a realization $\Theta=\theta$ is given, the entire hazard trajectory $h_\theta(t)$ is fixed, and the conditional distribution of the event time $T$ is fully determined.
Under this interpretation, observable information $X$ plays a purely informational role: it provides partial knowledge about the latent mechanism $\Theta$ but does not modify the event-time distribution once the mechanism is fixed.
This informational separation implies the conditional independence relation
\[
    T \perp X | \Theta .
\]
This relation is not introduced as an additional modeling assumption, but follows directly from the definition of $\Theta$ as a hazard-generating mechanism.
If $\Theta$ indeed captures all individual-level structure governing risk, then no observable covariate should exert further influence on $T$ beyond what is already encoded in $\Theta$.
This does not preclude covariates from being highly predictive of risk, but locates their role entirely in refining information about the underlying mechanism.

Importantly, mechanism sufficiency does not require the observable information $X$ to uniquely determine $\Theta$.
In most practical settings, $X$ is only partially informative, so that substantial heterogeneity in $\Theta$ typically remains within covariate-defined groups.
This residual uncertainty is precisely what gives rise to group-level survival and hazard functions as mixtures over latent mechanisms, rather than as individual-level risk trajectories.
From a technical perspective, mechanism sufficiency also plays a supporting role in the theoretical developments that follow.
By treating $\Theta$ as the complete hazard-generating object, this principle ensures that all stochastic variation in the event time $T$ is mediated through $\Theta$, which allows identifiability questions to be formulated entirely at the level of conditional mechanism distributions.
In particular, this separation is used in the proof of Theorem~2.3 to characterize the aggregation operator mapping individual mechanisms to population-level survival functions.

\paragraph{Information versus causation.}
Throughout this paper, covariates are treated as sources of information about latent individual hazard mechanisms rather than as causal drivers of those mechanisms.
This distinction is deliberate.
Conditioning on observed covariates $X$ reflects an operation of statistical conditioning rather than a statement about how interventions on $X$ would alter the underlying hazard mechanism.
Accordingly, the conditional distribution $\mathbb{P}(\Theta | X=x)$ encodes partial information about individual mechanisms, not a causal pathway from $X$ to $\Theta$.

This perspective is compatible with, but logically distinct from, causal interpretations of survival models.
Causal questions concern hypothetical interventions, such as whether manipulating a treatment or exposure would change the underlying hazard mechanism.
By contrast, the present framework focuses on the information structure induced by observation: what can be inferred about individual risk mechanisms from survival data and covariates as they are observed.
These perspectives coincide only under strong assumptions linking covariates to mechanisms through deterministic or near-deterministic relationships.
Making this distinction explicit clarifies the interpretation of observable hazards and survival functions.
Even when covariates are causally relevant, fitted hazards conditioned on $X$ represent survivor-weighted averages over heterogeneous mechanisms rather than intrinsic individual-level risks.
The latent hazard framework therefore does not deny the relevance of causal reasoning in survival analysis, but emphasizes that causal questions and identifiability questions operate at different analytical levels and should be kept conceptually distinct.

\subsubsection{Population and group-level survival}

We now show how classical survival and hazard functions arise as aggregated quantities under partial information.

\begin{theorem}[Representation of observable survival and hazard]
\label{thm:representation}
    Let $\Theta$ be an individual hazard mechanism with induced hazard trajectory
    $h_\Theta(t)$ and associated survival curve $S_\Theta(t)$.
    Let $X$ be an observable covariate interpreted as partial information about $\Theta$.
    Then the following statements hold.
    \begin{enumerate}[label=(\roman*)]  
        \item The population-level survival function
        \(
        S(t) = \mathbb P(T>t)
        \)
        admits the representation
        \[
        S(t) = \mathbb E\!\left[S_\Theta(t)\right],
        \]
        that is, it arises as the average of individual survival curves over the latent
        mechanism distribution.
        
        \item For any covariate value $x$ with $\mathbb P(X=x)>0$, the group-level survival function
        satisfies
        \[
        S(t | x)
        =
        \mathbb P(T>t | X=x)
        =
        \mathbb E\!\left[S_\Theta(t) | X=x\right].
        \]
        
        \item If $S(t | x)$ is differentiable, the associated group-level (observable) hazard is given by
        \[
        h_{\mathrm{obs}}(t | x)
        =
        -\frac{d}{dt}\log S(t | x)
        =
        \mathbb E\!\left[h_\Theta(t) | T \ge t,\, X=x\right].
        \]
    \end{enumerate}
\end{theorem}

The proof of Theorem~\ref{thm:representation} is provided in Appendix~\ref{proof:thm_representation}.
The theorem makes explicit that survival and hazard functions are not primitive individual-level objects.
Population-level and group-level survival functions arise as mixture averages over latent mechanisms, while observable hazards correspond to survivor-weighted posterior averages of individual hazard trajectories.
All individual-level heterogeneity is therefore filtered through the observable survival function, clarifying both the scope and the intrinsic limitations of hazard-based inference.

\subsection{Mechanism-Level Averages and Their Interpretation}

For conceptual clarity, it is useful to distinguish the observable hazard $h_{\mathrm{obs}}(t | x)$ from other summaries that may be defined at the mechanism level.
In particular, one may consider the conditional mechanism-level average
\begin{equation}\label{eq:mechanism-average}
    \bar h(t | x)
    =
    \mathbb E\!\left[h_\Theta(t) | X=x\right]
    =
    \int_{\mathcal M} h_\theta(t)\,\mathbb P_\Theta(d\theta | X=x),
\end{equation}
which averages individual hazard trajectories with respect to the \emph{conditional} distribution of latent mechanisms compatible with $x$.
While $\bar h(t | x)$ is mathematically well defined, it is not directly observable from survival data.
In general,
\[
    h_{\mathrm{obs}}(t | x) \neq \bar h(t | x),
\]
because the observable hazard involves additional conditioning on survival up to time $t$, as established in the previous subsection.
The discrepancy between these two quantities reflects the combined effects of unobserved heterogeneity in individual mechanisms and dynamic selection within the risk set.
Only in the degenerate case where the conditional distribution $\mathbb P_\Theta(\cdot | X=x)$ collapses to a point mass, that is, when the covariates uniquely determine the underlying mechanism, do the two hazards coincide.

This distinction is central to the interpretation of survival models.
Classical parametric, semiparametric, and machine-learning approaches estimate or approximate the observable hazard $h_{\mathrm{obs}}(t | x)$.
They do not, and cannot without additional structural assumptions, recover individual hazard trajectories or their mechanism-level averages.
Accordingly, interpretations that treat fitted hazards as individual-level risk functions implicitly impose strong and typically unacknowledged restrictions on the latent mechanism distribution.
Within the present framework, the observable hazard should therefore be viewed as the result of an aggregation operator acting on latent mechanisms,
\[
    \{h_\theta(\cdot),\, \mathbb P_\Theta(\cdot | X=x)\}
    \;\longmapsto\;
    h_{\mathrm{obs}}(\cdot | x),
\]
which integrates both structural heterogeneity across mechanisms and time-dependent survivor selection.
This perspective clarifies why individual-level hazard quantities and their mechanism-level averages are, in general, not identifiable from survival data, and motivates the identifiability analysis developed in the subsequent section.

\section{Identifiability Theory}
\label{sec:identifiability}

Building on the framework introduced in Section~\ref{sec:framework}, we examine the identifiability implications of modeling survival data through latent individual-level mechanisms.
The central object of interest is the individual hazard mechanism $\Theta$.
Identifiability in this setting, however, is not a property of $\Theta$ alone, but of the information structure linking $\Theta$, observable covariates $X$, and the event-time process.
As a consequence, what can be learned from survival data is fundamentally constrained by the way individual-level heterogeneity is aggregated into observable quantities.
These limitations do not hinge on censoring.
Even in the idealized setting where all event times are fully observed, the aggregation of heterogeneous individual mechanisms into population-level quantities already entails irreducible constraints on identifiability.

Accordingly, this section focuses on identifying which aspects of the latent mechanism distribution are, in principle, recoverable from survival data, and which limitations are unavoidable.
Rather than attempting to reconstruct individual survival curves, we clarify the scope of identifiability at the population level under partial information.
Specifically, we first analyze how the information structure restricts identifiability.
We then introduce a non-degeneracy principle under which these restrictions can be formalized.
Finally, we show how the resulting non-identifiability implies that the conditional distribution of mechanisms given covariates should be interpreted as an irreducible residual structure in survival analysis.

\subsection{Information Structure and the Nature of Identifiability}

Identifiability is fundamentally a property of the information structure rather than of any particular statistical model.
In statistical theory, what can be identified from data is determined by the mapping between latent quantities and observable variables, and by whether this mapping preserves sufficient information to distinguish distinct underlying states (see, e.g., \citep{Kagan1973,LehmannCasella1998}).
Questions of identifiability therefore precede modeling choices, estimation procedures, and assumptions on functional forms.
In survival analysis, this information structure is shaped not only by observed covariates, but also by censoring and by the aggregation of individual risk over time.

In applied survival analysis, the underlying object of interest is often taken to be an individual-level survival or hazard function, implicitly associated with an unobserved individual mechanism, particularly in settings where the goal is individual-level risk prediction.
However, survival data do not provide direct access to such mechanisms. 
Instead, information about individual survival behavior is mediated through a limited set of observables, typically consisting of the event time, censoring indicator, and covariates.
This mediation induces an intrinsic and irreversible loss of information, in the sense that multiple distinct individual mechanisms may give rise to the same observable distribution.

As a consequence, identifiability issues in survival analysis cannot be resolved solely by adopting more flexible models or richer parameterizations.
Even under arbitrarily flexible specifications, the observable data may be insufficient to distinguish between heterogeneous individual survival patterns.
The resulting limitations are therefore structural in nature and arise from the aggregation of individual-level behavior at the population level, rather than from model misspecification or insufficient expressiveness.

This perspective suggests that identifiability questions in survival analysis should be addressed at the level of information structure.
Rather than asking whether a particular model or parameterization is identifiable, a more fundamental question is which aspects of the underlying survival mechanisms are, in principle, recoverable from the available observables, and which are inherently lost.
In particular, it is essential to distinguish between quantities that are identifiable only at the population or group level and those that pertain to individual hazard mechanisms.
The remainder of this section builds on this viewpoint to characterize the scope and limits of identifiability under partial information.

\subsection{Non-degeneracy Assumptions and Identifiability Results}

Building on the framework introduced in Section~\ref{sec:framework}, we view survival data as arising from latent individual-level mechanisms that govern hazard trajectories and induce individual survival curves. 
Within this framework, the observable covariates $X$ are interpreted as partial information about the underlying mechanism $\Theta$, rather than as deterministic predictors of individual survival behavior. 
Consequently, the relationship between $\Theta$ and $X$ plays a central role in determining what aspects of individual heterogeneity can be learned from data. 
Therefore, a necessary conceptual prerequisite for interpreting survival quantities as reflecting individual-level risk is that the conditional distribution of mechanisms given covariates, $\mathbb{P}(\Theta | X)$, should not collapse to a degenerate point mass.
This non-degeneracy principle reflects the fact that covariates typically provide incomplete information about individual survival mechanisms: individuals with identical or similar covariate values may still exhibit substantial heterogeneity in their underlying hazard dynamics. 
In particular, we exclude settings in which the individual mechanism is almost surely a deterministic function of the covariates, i.e., $\Theta = g(X)$ almost surely.

This principle is especially natural in contexts where the objective is individual-level risk prediction. 
While covariates may be informative about survival outcomes, it is rarely plausible that they fully determine an individual hazard mechanism. 
Treating $X$ as partial rather than complete information about $\Theta$, therefore,  provides a more realistic representation of individual variability and aligns with the information structure implicit in survival data.
Importantly, the non-degeneracy principle should be understood as a structural modeling stance rather than a technical assumption imposed for mathematical convenience. 
It formalizes the idea that irreducible uncertainty remains at the individual level even after conditioning on observed covariates. 

To make the above non-degeneracy principle concrete, we introduce a technical formulation that is sufficient for deriving identifiability results. 
This formulation is not intended to be exhaustive or necessary; rather, it provides a transparent and tractable realization of the principle that covariates convey only partial information about individual mechanisms

\begin{assumption}[Conditional Local Richness]
\label{ass:local_richness}
    Fix a covariate value $x$ and consider the space of latent hazard mechanisms that may arise under the condition $X=x$. 
    We assume that, conditional on $X=x$, the latent mechanism $\Theta$ remains non-degenerate in the sense that non-trivial individual-level variability persists at the level of survival trajectories.
    Specifically, we assume that there exist a reference mechanism $\theta_0 \in \mathcal M$, a bounded measurable function $g:[0,\infty)\to\mathbb R$ that is not identically zero, a constant $\delta>0$, and a measurable mapping
    \[
        \varepsilon \longmapsto \theta(\varepsilon)\in\mathcal M,
        \qquad \varepsilon\in(-\delta,\delta),
    \]
    such that for all $\varepsilon\in(-\delta,\delta)$, the associated individual survival functions satisfy
    \[
        S_{\theta(\varepsilon)}(t) = S_{\theta_0}(t)\bigl(1+\varepsilon g(t)\bigr), \qquad \forall\, t\ge 0,
    \]
    and define valid survival functions for each $\varepsilon$.
\end{assumption}

This assumption formalizes a local richness property of the mechanism space: even after conditioning on the covariate value $X=x$, the space of admissible survival trajectories admits non-trivial perturbations. 
Consequently, the conditional distribution $\mathbb{P}(\Theta | X=x)$ does not collapse to a degenerate point mass.
It is important to emphasize that this assumption should be understood as one sufficient technical realization of the non-degeneracy principle introduced above, rather than as a fundamental requirement. 
We emphasize that the non-degeneracy principle itself does not assert that all non-degenerate mechanism spaces necessarily lead to non-identifiability.
Rather, the present assumption shows that once non-degeneracy admits even mild local variation in survival trajectories, aggregation inevitably destroys injectivity.
Alternative formulations that preserve the conditional variability of mechanisms given covariates would lead to analogous identifiability conclusions. 
The role of the present assumption is therefore to enable a precise characterization of identifiability consequences under partial information, rather than to impose a restrictive structural model.
In this sense, non-identifiability should be viewed as a consequence of allowing irreducible heterogeneity at the mechanism level, rather than as a failure of statistical modeling.

Under the conditional local richness assumption, the mapping from individual survival trajectories and their conditional distribution to the population-level survival function constitutes an aggregation operator that is inherently many-to-one.
The following theorem formalizes this observation by establishing a fundamental non-identifiability result at the individual level: the mapping from individual survival trajectories and their conditional distribution to the observable population-level survival function is inherently non-injective, and individual survival curves and hazard trajectories therefore cannot be uniquely recovered from population-level survival data, even when the mechanism space is fixed and the population survival function is fully known.

\begin{theorem}[Structural non-identifiability at the individual level]
\label{thm:nonid-individual}
    Fix a covariate value $x$ and let the individual hazard mechanism $\Theta$ take values in a measurable space $(\mathcal M,\mathcal B)$. 
    Suppose that the population-level survival function $S(t | x)$ is well defined on $[0,\infty)$. 
    Assume further that, conditional on $X=x$, the mechanism space $\mathcal M$ satisfies the conditional local richness assumption introduced above.
    Then the observation operator mapping individual-level survival trajectories to the population-level survival function,
    \[
        \mathcal O:\;
        \bigl(\{S_\theta(\cdot)\}_{\theta\in\mathcal M},\, \mathbb P_\Theta(\cdot | X=x)\bigr)
        \;\longmapsto\;
        S(t | x)
        \;=\;
        \int_{\mathcal M} S_\theta(t)\, \mathbb P_\Theta(d\theta | X=x),
    \]
    is non-injective at the level of survival paths.
    More precisely, even when the mechanism space $\mathcal M$ is fixed, there exist infinitely many distinct conditional distributions
    \[
        \mathbb P^{(k)}_\Theta(\cdot | X=x), \qquad k=1,2,\ldots,
    \]
    such that, for all $k$,
    \[
        S(t | x)
        =
        \int_{\mathcal M} S_\theta(t)\, \mathbb P^{(k)}_\Theta(d\theta | X=x),
        \qquad \forall\, t \ge 0,
    \]
    while the corresponding conditional distributions differ at the path level. Consequently, the induced collections of individual survival curves, and the associated individual hazard trajectories $h_\theta(t)$, are not uniquely identifiable from the observable population-level survival function $S(t | x)$.
\end{theorem}

The proof of Theorem~\ref{thm:nonid-individual} is provided in Appendix~\ref{proof:thm_nonid}.
This theorem shows that the lack of identifiability at the individual level is a structural consequence of the information aggregation inherent in survival data. 
Even when the population-level survival function $S(t| x)$ is fully known and the mechanism space $\mathcal M$ is fixed, the mapping from individual survival trajectories and their conditional distribution to the observable population-level survival curve is not invertible. 
As a result, individual survival curves and hazard trajectories cannot be uniquely recovered from $S(t| x)$.

Importantly, this non-identifiability does not arise from model misspecification or insufficient flexibility, but from the fact that population-level survival functions represent averages over heterogeneous individual mechanisms. 
Different configurations of individual-level survival paths may therefore induce identical population-level behavior, rendering them observationally indistinguishable. 
From this perspective, the population survival function $S(t | x)$ should be understood as an aggregated object rather than as a direct proxy for any individual survival trajectory.
The only situation in which this aggregation-induced non-identifiability disappears is the degenerate case where conditioning on the observable covariates eliminates all residual heterogeneity in the latent mechanisms. 
If the conditional distribution $\mathbb{P}(\Theta | X=x)$ collapses to a point mass at a single mechanism, the aggregation operator becomes trivial and the population-level survival function coincides with the survival trajectory induced by that mechanism. 
Such perfect information scenarios represent knife-edge boundary cases in which individual survival curves are, in principle, identifiable. 
In realistic applications, however, observable covariates rarely exhaust individual-level heterogeneity, and population survival functions remain genuine aggregates over latent mechanisms.

This result clarifies the scope of what can and cannot be learned from survival data.
Although population-level quantities and summaries are identifiable under appropriate conditions, individual-level survival mechanisms remain fundamentally underdetermined without additional structural assumptions. 
In the next subsection, we build on this insight by interpreting the conditional distribution of mechanisms given covariates as a residual structure in survival analysis.

\subsection{Conditional Mechanism Distributions as Residual Structure in Survival Models}

A central feature of survival analysis is the presence of multiple sources of randomness operating at different levels of the data-generating process.
At a minimum, randomness arises from the event-time mechanism itself, reflecting the stochastic nature of survival outcomes even when an individual hazard mechanism is fixed.
This source of randomness is explicitly represented in virtually all survival modeling frameworks through hazard functions, survival functions, or counting process formulations.
In addition to event-time variability, survival data also reflect heterogeneity across individuals in their underlying hazard-generating mechanisms.
Even after conditioning on observed covariates, individuals may differ in their latent hazard trajectories in ways that are not directly observable.
Unlike event-time randomness, however, this mechanism-level heterogeneity is typically not represented as an explicit stochastic component of the model.
Instead, it is often absorbed into population-level quantities, constrained through low-dimensional parametric structures, or treated implicitly through modeling assumptions.
The latent hazard framework introduced in Section~\ref{sec:framework} makes this additional layer of randomness explicit by treating individual hazard mechanisms as latent random objects and by separating mechanism-level heterogeneity from event-time stochasticity.

From an information-theoretic perspective, this residual heterogeneity reflects uncertainty that cannot be eliminated by conditioning on $X$ alone.
The identifiability results in the previous subsection show that this mechanism-level uncertainty cannot, in general, be resolved from population-level survival data.
Distinct configurations of individual survival trajectories may induce identical observable survival functions, implying that some aspects of the underlying data-generating process are inherently unidentifiable.
This observation suggests that any survival model must, either explicitly or implicitly, take a stance on how such irreducible uncertainty is represented or absorbed.
\paragraph{Two complementary perspectives on non-identifiability.}
The non-identifiability of the conditional mechanism distribution $\mathbb{P}(\Theta | X=x)$ can be understood from two complementary but distinct perspectives. 
First, at the level of the data-generating process, this non-identifiability reflects an inability to recover the structural relationship between latent individual mechanisms and observable covariates from survival data alone.
In this sense, $\mathbb{P}(\Theta | X=x)$ plays a role analogous to the unknown structural form of the data-generating process in regression analysis: just as the joint distribution of $(Y,X)$ does not determine whether the true relationship is linear, additive, or nonlinear without additional assumptions, survival data do not uniquely determine the conditional distribution of latent hazard mechanisms.
Modeling choices such as proportional hazards, accelerated failure time, or frailty specifications therefore correspond to structural assumptions imposed at this level, rather than quantities identified from the data.
Second, from a model-based perspective, once a particular structural form has been adopted, the conditional mechanism distribution captures the residual heterogeneity that remains unexplained by the observed covariates within the chosen model class.
In this sense, $\mathbb{P}(\Theta | X=x)$ plays a role analogous to a residual component: it represents individual-level variability that cannot be attributed to $X$ under the imposed structural assumptions.
Importantly, this residual interpretation is conditional on the modeling framework and should not be conflated with the more fundamental non-identifiability that arises at the data-generating level.

Taken together, these two perspectives clarify the role of conditional mechanism distributions in survival analysis.
Their non-identifiability is not a modeling defect, nor a consequence of insufficient flexibility, but a structural feature of the information available in survival data.
Different survival models may parametrize, restrict, or approximate $\mathbb{P}(\Theta | X=x)$ in different ways, but none can eliminate its presence entirely unless individual survival mechanisms are fully determined by observed covariates.

\section{Classical Survival Models under a Latent Hazard Framework}
\label{sec:classical}

In the preceding sections, we have argued that individual survival outcomes arise from latent, individual-specific hazard mechanisms and that observable survival quantities, such as population-level or group-level survival and hazard functions, are obtained through aggregation over these unobserved mechanisms. 
A central consequence of this representation is that, in the presence of covariate information \(X\), the conditional distribution of individual mechanisms,
\[
    \mathbb{P}(\Theta | X = x),
\]
is, in general, not identifiable from observed survival data.
This observation has fundamental implications for survival analysis. 
It implies that any statistical model for survival data must, either explicitly or implicitly, take a stance on how this conditional mechanism distribution is handled.
Some models bypass it altogether by operating exclusively at the population or group level; others impose restrictive structural assumptions that render the problem tractable; still others attempt to approximate or discretize the mechanism space in order to capture heterogeneity or latent subtypes. 
These choices are rarely articulated in mechanism-level terms, yet they determine the class of individual heterogeneity that a model is capable of representing.

The purpose of this section is therefore not to review classical survival models in terms of their regression forms or estimation procedures, but to reinterpret them within a framework based on mechanism conditioning.
Specifically, we examine how different model classes handle the unresolved variability encoded in \(\mathbb{P}(\Theta | X = x)\), which we refer to as the conditional mechanism distribution or, by analogy with regression analysis, the residual mechanism component. 
By organizing classical models according to their treatment of this object, we aim to clarify their conceptual scope, their implicit assumptions, and the sources of their structural limitations.

\subsection{Marginal Hazard Modeling and the Cox Framework}
\label{sec:cox}

The Cox's proportional hazards model is a cornerstone of modern survival analysis.
Within the latent hazard framework developed in this paper, it can be understood as adopting a specific strategy for handling the conditional mechanism distribution \(\mathbb{P}(\Theta | X = x)\).
Rather than modeling this distribution explicitly, the Cox model bypasses the mechanism level altogether by imposing structure only on an aggregated, population-level quantity.
In the presence of latent individual hazard mechanisms, the object of direct modeling is the observable or group-level hazard, given in our framework by the survivor-weighted posterior expectation of individual hazard trajectories conditional on remaining at risk.
This observable hazard already integrates over the conditional mechanism distribution \(\mathbb{P}(\Theta | X = x)\) as well as the dynamic selection induced by survival up to time \(t\).
The Cox model specifies that this quantity factorizes as
\[
    h_{\mathrm{obs}}(t | x) = h_0(t)\exp(\beta^\top x),
\]
thereby sidestepping the problem of modeling or identifying \(\mathbb{P}(\Theta | X = x)\) at the individual level.
Under this formulation, the hazard ratio between two covariate values \(x_1\) and
\(x_2\),
\[
    \frac{h_{\mathrm{obs}}(t | x_1)}{h_{\mathrm{obs}}(t | x_2)}
    =
    \exp\!\big(\beta^\top (x_1 - x_2)\big),
\]
should be interpreted as a comparison between two \emph{survivor-weighted posterior mixtures} of individual hazard mechanisms.
It does not represent a comparison of intrinsic individual-level risks, but rather reflects how covariates reweight the conditional mechanism distribution among individuals who remain in the risk set at time \(t\).
Proportional hazards is therefore a statement about the relative behavior of observable mixtures, not about proportionality of individual hazard trajectories.

This marginal modeling strategy has direct consequences at the mechanism level.
Compatibility with proportional hazards at the observable level severely restricts the admissible forms of conditional mechanism heterogeneity.
In particular, genuine variation in hazard shape, timing, or modality across individuals cannot be represented, as all such variation must be absorbed into risk set selection and aggregation.
From the same perspective, the baseline hazard \(h_0(t)\) should be viewed as a statistical device introduced to represent the observable hazard, rather than as the baseline of an underlying individual hazard mechanism.
Its form is not uniquely determined by the data and carries no direct mechanistic
interpretation.
A formal characterization of the mechanism-level restrictions implied by proportional hazards, including sufficient conditions under which observable proportionality can arise from latent mechanisms, is provided in Appendix~\ref{app:cox-mechanism}.

\subsection{Low-Dimensional Mechanism Variation: Frailty Models}

Frailty models are commonly introduced as extensions of the Cox proportional hazards model designed to account for unobserved heterogeneity among individuals \cite{vaupel1979impact, hougaard1995frailty}.
From the perspective adopted in this paper, their defining feature is not the introduction of random effects per se, but the particular way in which they restrict the conditional mechanism distribution \(\mathbb{P}(\Theta | X = x)\).
In a typical frailty formulation, the individual hazard function is written as
\[
    h(t | X = x, Z = z) = z \, h_0(t)\exp(\beta^\top x),
\]
where \(Z\) is a non-negative random variable, independent of \(X\), commonly referred to as the frailty term.
At the mechanism level, this specification amounts to assuming that all individual hazard mechanisms share a common baseline shape and differ only through a multiplicative, time-invariant scaling factor.
Consequently, the conditional mechanism distribution is restricted to a one-dimensional family of proportional hazards indexed by \(Z\).

This explicit acknowledgment of residual mechanism variability represents a first step beyond the Cox model, but it comes with severe structural constraints.
Because heterogeneity is confined to a single scalar dimension, frailty models cannot represent differences in hazard shape, timing, or modality; in particular, mechanisms corresponding to early- versus late-onset risk, multi-phase hazard trajectories, or non-proportional effects lie outside the admissible mechanism space.
Moreover, this low-dimensional representation does not resolve the fundamental non-identifiability of \(\mathbb{P}(\Theta | X = x)\), but instead renders the problem tractable by imposing a strong assumption: conditional on covariates, individual mechanisms vary only through a time-invariant random effect.
From this viewpoint, frailty models should be understood as approximations that trade expressive power for interpretability and estimability, and whose primary role is to adjust population-level inference in the presence of unobserved heterogeneity rather than to recover individual-level mechanisms or identify latent subtypes.
From this perspective, frailty models do not represent a general framework for individual heterogeneity, but a specific parametric closure of the conditional mechanism distribution.
The latent hazard framework adopted here reverses this logic: the mechanism distribution is taken as primitive, and classical models, including frailty specifications, are understood as particular restrictions imposed for tractability rather than as solutions to identifiability.

\subsection{Residual Mechanisms and Time-Scale Transformations: AFT Models}
\label{sec:aft}

In the classical literature \cite{kalbfleisch2002statistical}, AFT models are most commonly introduced through a regression formulation on the logarithmic time scale,
\[
    \log T = \mu(X) + \varepsilon,
\]
or, equivalently, through a time–scaling representation of the survival function,
\[
    S(t | X) = S_0\!\left(t\,e^{-\mu(X)}\right),
\]
which corresponds to a systematic rescaling of the time axis \(t \mapsto t/a(X)\), with \(a(X)=e^{\mu(X)}\).
In traditional treatments, the regression form and the distribution of the error term \(\varepsilon\) are often regarded as the primary modeling components.

Within the mechanism-based framework developed in this paper, we adopt a different perspective.
Rather than viewing AFT models as regression models for survival times, we interpret them as imposing a structural constraint on the conditional mechanism distribution \(\mathbb{P}(\Theta | X = x)\).
Specifically, AFT models assume the existence of a single reference hazard mechanism from which all individual hazard mechanisms are generated through time–scale transformations.
Conditional on \(X=x\), residual heterogeneity among individual mechanisms is therefore confined to acceleration or deceleration along a common hazard shape.
Under this interpretation, the error term \(\varepsilon\) in the classical AFT representation is not an independent modeling choice, but an induced quantity reflecting the survival time generated by the reference mechanism under a logarithmic time transformation.
Consequently, the conventional choice of an error distribution implicitly corresponds to a choice of the reference hazard shape, rather than introducing an additional degree of freedom in the conditional mechanism distribution.

This formulation makes explicit that AFT models impose a strong compression of
\(\mathbb{P}(\Theta | X = x)\).
All conditional heterogeneity is restricted to a one-dimensional time–scaling factor, precluding genuine variation in hazard shape, timing, or modality across
individuals.
As in the case of frailty models, this structural simplicity confers interpretability and stability, but it does not resolve the fundamental non-identifiability of the conditional mechanism distribution.
A formal mechanism-based generative characterization of AFT models, together with its implications for the structure of \(\mathbb{P}(\Theta | X)\), is provided in Appendix~\ref{app:aft-mechanism}.

\subsection{Discrete Mechanism Representations and Survival Clustering}

Survival clustering and latent class approaches are motivated by the observation that heterogeneous survival patterns may reflect the presence of distinct, unobserved subtypes within a population.
Within the mechanism-based framework developed here, these methods can be understood as adopting a specific and highly structured strategy for handling the conditional mechanism distribution $\mathbb{P}(\Theta | X = x)$.

Recent survival clustering methods, such as the Survival Cluster Analysis framework \cite{chapfuwa2020survival}, exemplify this strategy by introducing a discrete latent structure into time-to-event models.
Rather than aiming to identify individual-level mechanisms, such approaches impose a structured approximation to $\mathbb{P}(\Theta | X = x)$ in order to obtain tractable and interpretable representations of heterogeneity.
At the mechanism level, survival clustering models typically posit that the conditional distribution of individual mechanisms admits a finite discrete representation,
\[
    \mathbb{P}(\Theta | X = x) = \sum_{k=1}^K \pi_k(x)\delta_{\Theta_k},
\]
where each $\Theta_k$ corresponds to a distinct hazard mechanism or subtype,  $\delta_{\Theta_k}$ denotes a Dirac measure concentrated at $\Theta_k$, and $\pi_k(x)$ denotes covariate-dependent mixing weights.
In many practical implementations, each component may correspond to a class-specific family of hazard models rather than a single fixed hazard; here we adopt an abstracted mechanism-level view in which each component is represented by a representative mechanism.
Under this assumption, heterogeneity is not modeled as continuous variation around a common mechanism, but rather as arising from a finite collection of qualitatively distinct mechanisms.

From this perspective, survival clustering replaces the unresolved variability in \(\mathbb{P}(\Theta | X = x)\) with a strong discretization assumption. 
The conditional mechanism distribution is no longer treated as an unknown and potentially complex object, but is approximated by a finite mixture supported on a small number of latent classes. 
This discretization renders the problem tractable and provides a direct interpretation in terms of subtypes, but it does so by substantially restricting the admissible mechanism space.

It is important to emphasize that this strategy does not circumvent the fundamental non-identifiability of \(\mathbb{P}(\Theta | X = x)\). 
The discrete representation is not learned from the data in an unconstrained sense; rather, it is imposed as a modeling assumption. 
In particular, the assignment of individuals to latent classes, as well as the interpretation of these classes as distinct mechanisms, relies on the assumption that the observed survival and covariate information is sufficient to support such a discretization. 
In general, this assumption cannot be justified without additional structural restrictions.

Viewed through the lens of mechanism conditioning, survival clustering methods are therefore best understood as approximation schemes rather than identification procedures. 
They offer a convenient and interpretable way to summarize heterogeneity by partitioning the mechanism space into a finite number of representative components, but this comes at the cost of introducing strong and often implicit assumptions about the nature of the conditional mechanism distribution. 
Their primary contribution lies in providing a low-complexity representation of heterogeneity, rather than in recovering individual-level mechanisms or uncovering uniquely identifiable subtypes.

\section{Discussion and Concluding Remarks}
\label{sec:discussion}

This paper was motivated by a persistent ambiguity in survival analysis concerning the relationship between individual-level risk and population-level observable quantities.
By introducing a latent hazard framework, we made explicit the information structure linking individual hazard-generating mechanisms, observable covariates, and survival outcomes.
Within this framework, classical survival quantities were shown to arise as aggregated objects, obtained by averaging over latent individual heterogeneity.
This perspective allows a unified treatment of identifiability and provides a conceptual basis for reinterpreting a broad class of survival models.

A central conclusion of the paper is that individual-level hazard trajectories are not identifiable from survival data under partial information.
More strongly, even the conditional distribution of individual mechanisms given covariates, $\mathbb{P}(\Theta | X=x)$, is generally not identifiable.
This lack of identifiability is structural rather than technical: it arises from the aggregation inherent in survival data and persists regardless of model flexibility or estimation strategy.
What is identifiable from the data are population-level or group-level survival and hazard functions, which summarize the behavior of heterogeneous individuals but do not correspond to any intrinsic individual-level risk quantity except in degenerate cases.

\paragraph{Implications for interpretation and prediction.}

These results have direct implications for the interpretation of fitted survival models.
In particular, observable hazards should not be interpreted as individual risk functions.
Even when conditioned on covariates, fitted hazards represent survivor-weighted averages over latent individual mechanisms.
Interpreting such quantities as characterizing individual risk implicitly assumes that the conditional mechanism distribution collapses to a point, an assumption that is rarely justified in practice.
The latent hazard framework clarifies that many common interpretations of model output rely on strong and often unacknowledged structural assumptions.

From this perspective, a key implication concerns the scope of interpretation in inferential settings.
Quantities derived from fitted survival models may support valid population-level inference within their modeling assumptions, but they do not automatically admit individual-level or mechanistic interpretations.
Such interpretations require additional structural assumptions linking the observable representation to an underlying hazard-generating mechanism.
When these assumptions are not made explicit, inferential conclusions should be understood as pertaining to aggregated or projected quantities, rather than to intrinsic individual risk processes.
Clarifying this distinction does not undermine classical inferential methods, but delineates the domain within which their interpretations remain conceptually coherent.

These interpretational issues become more severe in predictive settings.
In many contemporary applications, survival models are evaluated primarily through predictive performance measures, such as discrimination or calibration.
While such criteria are appropriate for assessing predictive accuracy, they do not by themselves justify mechanistic, causal, or individual-level interpretations of the resulting risk scores.
Within the latent hazard framework, predictive success may arise from stable population-level associations, even when no identifiable individual-level hazard mechanism is recovered.
Without a clear separation between predictive objectives and inferential interpretation, there is a risk that predictive performance is overstated as evidence of understanding underlying risk processes.
Making these distinctions explicit is therefore essential when survival models are deployed in modern predictive and machine-learning-based contexts.

\paragraph{Implications for model choice and methodological development.}

The non-identifiability of $\mathbb{P}(\Theta | X=x)$ also reframes the role of model choice in survival analysis.
Different modeling strategies correspond to different ways of handling this unresolved structure.
Some approaches, such as proportional hazards models, bypass the mechanism level by imposing structure directly on observable hazards.
Others, such as frailty or accelerated failure time models, restrict conditional heterogeneity to low-dimensional families.
Latent class and survival clustering methods discretize the mechanism space. 
From this perspective, model selection is not merely a question of predictive performance or goodness of fit, but a choice about which structural assumptions one is willing to impose on an intrinsically unidentifiable object.
%
%

These considerations are particularly relevant for individual-level risk prediction and subgroup discovery.
The results of this paper suggest that individual risk prediction is not primarily an algorithmic problem, but a structural one: meaningful individual
predictions require strong assumptions about the form of the conditional mechanism distribution.
Similarly, survival clustering methods should be understood as approximation or summarization schemes rather than as procedures that recover uniquely identifiable subtypes.
Their outputs depend fundamentally on the imposed discretization of an underlying distribution that is not identifiable from the data.
%

Finally, the latent hazard framework points to a direction for future methodological development.
Rather than seeking increasingly flexible regressions for observable hazards, methodological innovation may lie in developing principled ways to constrain, approximate, or interpret the conditional mechanism distribution.
This includes explicit modeling of mechanism heterogeneity, incorporation of scientific structure, or representation-based approaches that acknowledge the limits of identifiability.
Such developments would not eliminate the fundamental non-identifiability established here, but could lead to models whose assumptions are clearer and whose interpretations are better aligned with the information contained in survival data\\

In conclusion, the contribution of this paper is not the proposal of a new survival model, but a clarification of what survival analysis can and cannot identify under partial information.
By making the latent structure explicit, the framework developed here provides a coherent interpretation of classical models and highlights the structural assumptions underlying individual-level inference.
We hope that this perspective will help align methodological development and applied practice with the intrinsic limits of survival data, and encourage the development of survival models whose assumptions and interpretations are explicitly grounded in the available information structure.

\section*{Acknowledgements}
The author thanks Rolf Larsson for helpful comments on an earlier version of this manuscript.
This work was funded by an unrestricted grant from the Swedish Research Council (2024-02846).


\bibliographystyle{plainnat}
\bibliography{ref} 

@article{cox1972regression,
  title        = {Regression models and life-tables},
  author       = {Cox, D. R.},
  journal      = {Journal of the Royal Statistical Society: Series B (Methodological)},
  volume       = {34},
  number       = {2},
  pages        = {187--220},
  year         = {1972}
}

@book{kalbfleisch2002statistical,
  title     = {The Statistical Analysis of Failure Time Data},
  author    = {Kalbfleisch, John D. and Prentice, Ross L.},
  edition   = {2nd},
  year      = {2002},
  publisher = {John Wiley \& Sons},
  address   = {Hoboken, NJ},
  note      = {Comprehensive reference on survival models and hazard shape classification.}
}

@article{Buckley1979AFT,
  title   = {Linear Regression with Censored Data}, 
  author  = {Buckley, John and James, Ian},
  journal = {Biometrika},
  volume  = {66},
  number  = {3},
  pages   = {429--436},
  year    = {1979},
  doi     = {10.1093/biomet/66.3.429}
}

@article{Ishwaran2008RSF,
  title   = {Random Survival Forests},
  author  = {Ishwaran, Hemant and Kogalur, Udaya B. and Blackstone, Eugene H. and Lauer, Michael S.},
  journal = {The Annals of Applied Statistics},
  volume  = {2},
  number  = {3},
  pages   = {841--860},
  year    = {2008},
  doi     = {10.1214/08-AOAS169}
}

@article{Katzman2018,
  author  = {Katzman, Jared L. and Shaham, Uri and Bates, Jonathan and Cloninger, Alexander and Jiang, Tingting and Kluger, Yuval},
  title   = {DeepSurv: Personalized Treatment Recommender System Using a Cox Proportional Hazards Deep Neural Network},
  journal = {BMC Medical Research Methodology},
  volume  = {18},
  number  = {1},
  pages   = {24},
  year    = {2018}
}

@inproceedings{lee2018deephit,
  title={Deephit: A deep learning approach to survival analysis with competing risks},
  author={Lee, Changhee and Zame, William and Yoon, Jinsung and Van Der Schaar, Mihaela},
  booktitle={Proceedings of the AAAI conference on artificial intelligence},
  volume={32},
  year={2018}
}

@article{vaupel1979impact,
  title={The impact of heterogeneity in individual frailty on the dynamics of mortality},
  author={Vaupel, James W and Manton, Kenneth G and Stallard, Eric},
  journal={Demography},
  volume={16},
  number={3},
  pages={439--454},
  year={1979},
  publisher={Springer}
}

@article{hougaard1995frailty,
  title={Frailty models for survival data},
  author={Hougaard, Philip},
  journal={Lifetime data analysis},
  volume={1},
  number={3},
  pages={255--273},
  year={1995},
  publisher={Springer}
}

@article{Wang2021Survey,
  title   = {Machine Learning Methods for Survival Analysis: A Survey},
  author  = {Wang, Peng and Li, Yan and Reddy, Chandan K.},
  journal = {ACM Computing Surveys},
  volume  = {54},
  number  = {4},
  pages   = {1--36},
  year    = {2021},
  doi     = {10.1145/3453140}
}

@article{wiegrebe2024deep,
  title={Deep learning for survival analysis: a review},
  author={Wiegrebe, Simon and Kopper, Philipp and Sonabend, Raphael and Bischl, Bernd and Bender, Andreas},
  journal={Artificial Intelligence Review},
  volume={57},
  number={3},
  pages={65},
  year={2024},
  publisher={Springer}
}

@article{proust2014joint,
  title={Joint latent class models for longitudinal and time-to-event data: a review},
  author={Proust-Lima, C{\'e}cile and S{\'e}ne, Mb{\'e}ry and Taylor, Jeremy MG and Jacqmin-Gadda, H{\'e}l{\`e}ne},
  journal={Statistical methods in medical research},
  volume={23},
  number={1},
  pages={74--90},
  year={2014},
  publisher={SAGE Publications Sage UK: London, England}
}

@article{lin2002latent,
  title={Latent class models for joint analysis of longitudinal biomarker and event process data: application to longitudinal prostate-specific antigen readings and prostate cancer},
  author={Lin, Haiqun and Turnbull, Bruce W and McCulloch, Charles E and Slate, Elizabeth H},
  journal={Journal of the American Statistical Association},
  volume={97},
  number={457},
  pages={53--65},
  year={2002},
  publisher={Taylor \& Francis}
}

@inproceedings{chapfuwa2020survival,
  title={Survival cluster analysis},
  author={Chapfuwa, Paidamoyo and Li, Chunyuan and Mehta, Nikhil and Carin, Lawrence and Henao, Ricardo},
  booktitle={Proceedings of the ACM Conference on Health, Inference, and Learning},
  pages={60--68},
  year={2020}
}

@article{martinussen2013collapsibility,
  title={On collapsibility and confounding bias in Cox and Aalen regression models},
  author={Martinussen, Torben and Vansteelandt, Stijn},
  journal={Lifetime data analysis},
  volume={19},
  number={3},
  pages={279--296},
  year={2013},
  publisher={Springer}
}

@article{Grambsch2017,
  author  = {Grambsch, Patricia M.},
  title   = {Goodness-of-Fit and Diagnostics for Proportional Hazards Regression Models},
  journal = {Cancer Treatment and Research},
  volume  = {170},
  pages   = {221--240},
  year    = {2017}
}

@article{dumas2025hazard,
  title={How hazard ratios can mislead and why it matters in practice},
  author={Dumas, Elise and Stensrud, Mats J},
  journal={European Journal of Epidemiology},
  pages={1--7},
  year={2025},
  publisher={Springer}
}

@book{aalen2008survival,
  title={Survival and event history analysis: a process point of view},
  author={Aalen, Odd O and Borgan, {\O}rnulf and Gjessing, H{\aa}kon K},
  year={2008},
  publisher={Springer}
}

@book{Kagan1973,
  author    = {Kagan, A. M. and Linnik, Yu. V. and Rao, C. R.},
  title     = {Characterization Problems in Mathematical Statistics},
  publisher = {John Wiley \& Sons},
  year      = {1973},
  address   = {New York}
}

@book{LehmannCasella1998,
  author    = {Lehmann, E. L. and Casella, G.},
  title     = {Theory of Point Estimation},
  edition   = {2},
  publisher = {Springer},
  year      = {1998},
  address   = {New York}
}

@article{austin2020graphical,
  title={Graphical calibration curves and the integrated calibration index (ICI) for survival models},
  author={Austin, Peter C and Harrell Jr, Frank E and van Klaveren, David},
  journal={Statistics in Medicine},
  volume={39},
  number={21},
  pages={2714--2742},
  year={2020},
  publisher={Wiley Online Library}
}

@article{steyerberg2010assessing,
  title={Assessing the performance of prediction models: a framework for traditional and novel measures},
  author={Steyerberg, Ewout W and Vickers, Andrew J and Cook, Nancy R and Gerds, Thomas and Gonen, Mithat and Obuchowski, Nancy and Pencina, Michael J and Kattan, Michael W},
  journal={Epidemiology},
  volume={21},
  number={1},
  pages={128--138},
  year={2010},
  publisher={LWW}
}


\appendix

\section{Proof of Theorem \ref{thm:representation}}
\label{proof:thm_representation}

\begin{proof}
    We prove each statement in turn.
    \medskip
    \noindent\textit{(i) Population survival.}
    By definition,
    \[
        S(t)=\mathbb P(T>t)=\mathbb E[\mathbf 1_{\{T>t\}}].
    \]
    Applying iterated expectation with respect to $\Theta$ yields
    \[
        \mathbb E[\mathbf 1_{\{T>t\}}] = \mathbb E\!\left[\mathbb E\!\left[\mathbf 1_{\{T>t\}}| \Theta\right] \right].
    \]
    Conditional on $\Theta=\theta$, the event time distribution is determined by the hazard $h_\theta$, and hence
    \[
        \mathbb E[\mathbf 1_{\{T>t\}}| \Theta=\theta] = \mathbb P(T>t| \Theta=\theta) = S_\theta(t).
    \]
    Therefore $S(t)=\mathbb E[S_\Theta(t)]$.
    
    \medskip
    \noindent\textit{(ii) Group-level survival.}
    For any $x$ with $\mathbb P(X=x)>0$,
    \[
        S(t| x) = \mathbb P(T>t| X=x) = \mathbb E[\mathbf 1_{\{T>t\}}| X=x].
    \]
    Using iterated conditional expectation with respect to $\Theta$ gives
    \[
        \mathbb E[\mathbf 1_{\{T>t\}}| X=x] = \mathbb E\!\left[ \mathbb E[\mathbf 1_{\{T>t\}}| \Theta = \theta, X=x] | X=x \right].
    \]
    By mechanism sufficiency, the conditional distribution of $T$ given $\Theta$ does not depend on $X$, so that
    \[
        \mathbb E[\mathbf 1_{\{T>t\}}| \Theta=\theta,X=x] = \mathbb P(T>t| \Theta=\theta) = S_\theta(t).
    \]
    Hence
    \[
        S(t| x)=\mathbb E[S_\Theta(t)| X=x].
    \]
    
    \medskip
    \noindent\textit{(iii) Observable hazard.}
    Assume that differentiation and conditional expectation may be interchanged.
    Differentiating the group-level survival function yields
    \[
        \frac{d}{dt}S(t| x) =
        \mathbb E\!\left[\frac{d}{dt}S_\Theta(t)| X=x\right] =
        -\mathbb E\!\left[h_\Theta(t)S_\Theta(t)| X=x\right],
    \]
    since $\frac{d}{dt}S_\theta(t)=-h_\theta(t)S_\theta(t)$ for each $\theta$.
    Therefore,
    \[
        h_{\mathrm{obs}}(t| x) =
        -\frac{d}{dt}\log S(t| x) =
        \frac{\mathbb E\!\left[h_\Theta(t)S_\Theta(t)| X=x\right]}
             {\mathbb E\!\left[S_\Theta(t)| X=x\right]}.
    \]
    Finally, by mechanism sufficiency, noting that
    \[ 
        \mathbb P(\Theta\in d\theta| T\ge t,X=x) \propto S_\theta(t)\,\mathbb P_\Theta(d\theta| X=x)
        \footnote{
            Formally, $\mathbb P(\Theta\in d\theta | T\ge t, X=x)$ denotes the regular conditional probability measure on $(\mathcal M,\mathcal B)$. 
            The notation $d\theta$ is used in the sense of integration against measurable functions.
        },
    \]
    where the proportionality reflects a Bayesian update of the mechanism distribution: among individuals with covariate value $x$, mechanisms that are more likely to survive up to time $t$ are over-represented in the risk set at time $t$, the above ratio is exactly the conditional expectation
    \[
        h_{\mathrm{obs}}(t| x) = \mathbb E\!\left[h_\Theta(t)| T\ge t,X=x\right],
    \]
    which completes the proof.
\end{proof}

\section{Proof of Theorem \ref{thm:nonid-individual}}
\label{proof:thm_nonid}

\begin{proof}
    Fix a covariate value $x$ and denote the observed population-level survival function by $S(t):=S(t| x)$ for $t\ge 0$. 
    We construct infinitely many distinct conditional distributions $\mathbb P^{(k)}_\Theta(\cdot| X=x)$ that induce the same survival function $S(t)$ on a fixed mechanism space $(\mathcal M,\mathcal B)$.
    
    Since $S(\cdot| x)$ is well defined and we study the non-identifiability of $\mathbb P_\Theta(\cdot| X=x)$ on a fixed mechanism space, let $\mu_0$ be a reference conditional distribution on $(\mathcal M,\mathcal B)$ that induces $S(\cdot| x)$ under the observation operator, i.e.,
    \begin{equation}
    \label{eq:baseline-mu0}
        S(t)=\int_{\mathcal M} S_\theta(t)\,\mu_0(d\theta),\qquad \forall\, t\ge 0.
    \end{equation}
    In modeling terms, $\mu_0$ may be interpreted as a candidate conditional mechanism distribution $\mathbb P_\Theta(\cdot| X=x)$.
    
    By the conditional local richness assumption under $X=x$, there exist a reference mechanism $\theta_0\in\mathcal M$, a bounded measurable function $g:[0,\infty)\to\mathbb R$ that is not identically zero, a constant $\delta>0$, and a measurable mapping $\varepsilon\mapsto\theta(\varepsilon)\in\mathcal M$ for $\varepsilon\in(-\delta,\delta)$ such that, for all $\varepsilon\in(-\delta,\delta)$,
    \begin{equation}
    \label{eq:local-variation}
        S_{\theta(\varepsilon)}(t)=S_{\theta_0}(t)\bigl(1+\varepsilon g(t)\bigr),\qquad \forall\, t\ge 0,
    \end{equation}
    and $S_{\theta(\varepsilon)}(\cdot)$ defines a valid survival function for each $\varepsilon$.
    We now redistribute mass in the conditional mechanism distribution along directions that are indistinguishable under the observation operator. Fix any $\alpha\in(0,1)$ and define
    \[
        \delta'=\min\left\{\delta,\frac{1-\alpha}{\alpha}\delta\right\}.
    \]
    For $\varepsilon\in(-\delta',\delta')$, let $\varepsilon'=-\frac{\alpha}{1-\alpha}\varepsilon$ and define the two-point mixture measure
    \begin{equation}
    \label{eq:mu-eps}
        \nu_\varepsilon := \alpha\,\delta_{\theta(\varepsilon)}+(1-\alpha)\,\delta_{\theta(\varepsilon')}.
    \end{equation}
    We then perform a local replacement of $\mu_0$ by defining
    \begin{equation}
    \label{eq:mu0-eps}
        \mu_\varepsilon := \mu_0-\eta\,\delta_{\theta_0}+\eta\,\nu_\varepsilon,
    \end{equation}
    where $\eta\in(0,1)$ is chosen such that $0<\eta\le \mu_0(\{\theta_0\})$, ensuring that $\mu_\varepsilon$ is a nonnegative measure with total mass one. 
    (If $\mu_0(\{\theta_0\})=0$, one may first perform an equivalent infinitesimal adjustment of $\mu_0$ so that it assigns positive mass to $\theta_0$; this does not affect \eqref{eq:baseline-mu0} and serves only as a technical convenience.)
    Combining \eqref{eq:baseline-mu0} and \eqref{eq:mu0-eps}, for any $t\ge 0$ we obtain
    \begin{align*}
        \int_{\mathcal M} S_\theta(t)\,\mu_\varepsilon(d\theta)
        &=
        \int_{\mathcal M} S_\theta(t)\,\mu_0(d\theta)
        -\eta S_{\theta_0}(t)
        +\eta\int_{\mathcal M} S_\theta(t)\,\nu_\varepsilon(d\theta)\\
        &=
        S(t)-\eta S_{\theta_0}(t)
        +\eta\Bigl[\alpha S_{\theta(\varepsilon)}(t)+(1-\alpha)S_{\theta(\varepsilon')}(t)\Bigr].
    \end{align*}
    Using \eqref{eq:local-variation}, we further obtain
    \begin{align*}
        \alpha S_{\theta(\varepsilon)}(t)+(1-\alpha)S_{\theta(\varepsilon')}(t)
        &=
        \alpha S_{\theta_0}(t)\bigl(1+\varepsilon g(t)\bigr)
        +(1-\alpha)S_{\theta_0}(t)\bigl(1+\varepsilon' g(t)\bigr)\\
        &=
        S_{\theta_0}(t)\Bigl[1+\bigl(\alpha\varepsilon+(1-\alpha)\varepsilon'\bigr)g(t)\Bigr].
    \end{align*}
    Since $\varepsilon'=-\frac{\alpha}{1-\alpha}\varepsilon$, we have
    \[
        \alpha\varepsilon+(1-\alpha)\varepsilon'=0,
    \]
    and therefore
    \[
        \alpha S_{\theta(\varepsilon)}(t)+(1-\alpha)S_{\theta(\varepsilon')}(t)=S_{\theta_0}(t).
    \]
    Substituting back yields
    \[
        \int_{\mathcal M} S_\theta(t)\,\mu_\varepsilon(d\theta)
        =
        S(t)-\eta S_{\theta_0}(t)+\eta S_{\theta_0}(t)
        =
        S(t),\qquad \forall\, t\ge 0.
    \]
    Thus, for each sufficiently small $\varepsilon$, the measure $\mu_\varepsilon$ induces the same population-level survival function $S(\cdot| x)$.
    
    Because $g\not\equiv 0$, there exists $t_\ast$ such that $g(t_\ast)\neq 0$. For $\varepsilon\neq \tilde\varepsilon$, equation \eqref{eq:local-variation} implies
    $S_{\theta(\varepsilon)}(t_\ast)\neq S_{\theta(\tilde\varepsilon)}(t_\ast)$, and hence $\theta(\varepsilon)\neq \theta(\tilde\varepsilon)$ in the sense of distinct survival trajectories. Consequently, the two-point mixtures $\nu_\varepsilon$ and $\nu_{\tilde\varepsilon}$ have different supports, and therefore $\mu_\varepsilon\neq \mu_{\tilde\varepsilon}$.
    By choosing a sequence $\varepsilon_k\downarrow 0$ with $\varepsilon_k\neq \varepsilon_\ell$ for $k\neq \ell$, we obtain infinitely many distinct conditional distributions $\{\mu_{\varepsilon_k}\}_{k\ge 1}$ that all induce the same population-level survival function $S(\cdot| x)$. Hence, the observation operator $\mathcal O$ admits infinitely many distinct preimages for the same image $S(\cdot| x)$ and is therefore non-injective at the path level. As a result, the conditional mechanism distribution $\mathbb P_\Theta(\cdot| X=x)$ cannot be uniquely identified from $S(\cdot| x)$. 
    
    Finally, in settings where individual hazard trajectories are defined by
    $h_\theta(t)=-\frac{d}{dt}\log S_\theta(t)$, the same argument implies that individual-level hazard paths are likewise not identifiable.
\end{proof}

\section{Mechanism-Level Implications of the Cox Model}
\label{app:cox-mechanism}

In this appendix, we provide a formal characterization of the mechanism-level restrictions implied by imposing a proportional hazards structure on the observable hazard.
This material supports the interpretation given in Section \ref{sec:cox} and is included here to separate formal arguments from the main conceptual discussion.

Let \(\Theta\) denote a latent individual hazard mechanism taking values in a (measurable) mechanism space \(\mathcal{M}\), and let \(\mathcal H(\Theta)=h_\Theta\) denote the induced individual hazard function.
For notational clarity, we consider a finite mechanism space \(\mathcal{M}=\{\theta_1,\ldots,\theta_K\}\) with corresponding hazard trajectories \(\{h_1(t),\ldots,h_K(t)\}\).
Given covariates \(X=x\), the observable hazard at time \(t\) can be written as
\[
    h_{\mathrm{obs}}(t | x)
    =
    \mathbb{E}\!\left[
        h_\Theta(t)
        \,|\,
        T \ge t,\, X = x
    \right]
    =
    \sum_{k=1}^K
    w_k(t,x)\,h_k(t),
\]
where
\[
    w_k(t,x)
    =
    \mathbb{P}(\Theta=\theta_k | T \ge t,\, X=x)
\]
denotes the survivor-weighted posterior probability of mechanism \(\theta_k\).
We now state a sufficient condition under which proportional hazards at the observable level can arise from latent mechanism heterogeneity.

\begin{proposition}[Mechanism-level implications of proportional hazards]
\label{prop:mechanismPH}
    Assume that the observable hazard satisfies
    \[
        h_{\mathrm{obs}}(t | x)
        =
        h_0(t)\exp(\beta^\top x)
        \qquad
        \text{for all } t \ge 0 \text{ and all } x,
    \]
    and that the posterior weights \(w_k(t,x)\) vary nontrivially with \(x\).
    Then there exists a nonnegative function \(h_\ast(t)\) and positive constants \(c_1,\ldots,c_K\) such that
    \[
        h_k(t) = c_k\,h_\ast(t),
        \qquad
        k=1,\ldots,K.
    \]
    That is, proportional hazards at the observable level are compatible with latent mechanism heterogeneity only if all admissible individual hazard trajectories share a common time shape and differ solely by a multiplicative scale factor.
\end{proposition}

Proposition~\ref{prop:mechanismPH} formalizes the sense in which proportional hazards severely restrict the conditional mechanism distribution \(\mathbb{P}(\Theta | X=x)\).
In particular, any heterogeneity in hazard shape, timing, or modality across individual mechanisms is incompatible with proportionality of the observable hazard.
For expository clarity, the proposition is stated for a finite mechanism space.
The restriction is not essential: analogous results hold for infinite or continuous mechanism spaces under mild regularity conditions, and reflect a structural consequence of proportional hazards rather than an artifact of discretization.

\paragraph{Proof of Proposition \ref{prop:mechanismPH}}
\begin{proof}
    For each covariate value \(x\), the observable hazard admits the survivor-weighted mixture representation
    \[
        h_{\mathrm{obs}}(t | x) = \sum_{k=1}^K w_k(t,x)\,h_k(t), \qquad t \ge 0,
    \]
    where \(w_k(t,x)=\mathbb P(\Theta=\theta_k | T\ge t, X=x)\).
    By assumption, the observable hazard satisfies the proportional hazards form
    \[
        h_{\mathrm{obs}}(t | x) = h_0(t)\exp(\beta^\top x) \qquad \text{for all } t \ge 0 \text{ and all } x .
    \]
    Fix covariate values \(x^{(1)},\ldots,x^{(K)}\) such that the corresponding weight vectors
    \[
        w(x^{(j)}) = \bigl(w_1(t,x^{(j)}),\ldots,w_K(t,x^{(j)})\bigr), \qquad j=1,\ldots,K,
    \]
    are linearly independent for at least one (and hence all) \(t\).
    This is possible by the assumption that the posterior weights vary nontrivially with \(x\).
    If the span of \(\{w(t,x):x\}\) has dimension \(m<K\), the argument can be restricted to an \(m\)-component effective submixture; the conclusion remains unchanged.
    
    Define the \(K\times K\) matrix
    \[
        W = \bigl(W_{jk}\bigr), \qquad W_{jk}=w_k(t,x^{(j)}),
    \]
    which is invertible by construction.
    For each fixed \(t\), collect the individual hazard values into the vector
    \[
        h(t)=\bigl(h_1(t),\ldots,h_K(t)\bigr)^\top,
    \]
    and the corresponding observable hazards into
    \[
        g(t)=\bigl(h_{\mathrm{obs}}(t| x^{(1)}),\ldots,h_{\mathrm{obs}}(t| x^{(K)})\bigr)^\top.
    \]
    By the mixture representation,
    \[
        g(t)=W\,h(t).
    \]
    On the other hand, proportional hazards implies
    \[
        h_{\mathrm{obs}}(t| x^{(j)})=h_0(t)\exp(\beta^\top x^{(j)}), \qquad j=1,\ldots,K,
    \]
    so that
    \[
    g(t)=h_0(t)\,c,
    \qquad
    c_j=\exp(\beta^\top x^{(j)}),
    \]
    where the vector \(c\) does not depend on \(t\).
    Combining the two expressions for \(g(t)\) yields
    \[
        W\,h(t)=h_0(t)\,c.
    \]
    Since \(W\) is invertible, we obtain
    \[
        h(t)=h_0(t)\,W^{-1}c.
    \]
    Thus there exists a fixed vector \(a=W^{-1}c\) such that, for all \(t\),
    \[
        h_k(t)=a_k\,h_0(t), \qquad k=1,\ldots,K.
    \]
    Setting \(h_\ast(t)=h_0(t)\) and \(c_k=a_k>0\) completes the proof.
\end{proof}

\paragraph{On the interpretation of the baseline hazard}

The proportional hazards representation
\[
    h_{\mathrm{obs}}(t | x) = h_0(t)\exp(\beta^\top x)
\]
introduces the baseline hazard \(h_0(t)\) as a component of a factorization of the observable hazard.
From the mechanism-level perspective adopted in this appendix, it is important to distinguish this statistical object from any underlying individual hazard mechanism.

Proposition~\ref{prop:mechanismPH} establishes that, under proportional hazards, all admissible individual hazard trajectories must share a common time shape \(h_\ast(t)\), differing only by a multiplicative scale factor.
This function \(h_\ast(t)\) can be interpreted as a \emph{mechanism-level baseline shape}, in the sense that it characterizes the unique hazard geometry compatible with proportionality at the observable level.
However, the existence of such a function does not imply its identifiability from survival data.

In contrast, the Cox baseline hazard \(h_0(t)\) is not an estimate of \(h_\ast(t)\), nor does it correspond to the hazard of a representative or baseline individual.
Rather, \(h_0(t)\) is a modeling device arising from marginal hazard modeling and from the chosen factorization of the observable hazard.
Its form is not uniquely determined by the data, as different choices of \(h_0(t)\) may lead to the same partial likelihood.
This distinction clarifies why baseline hazards in the Cox model carry no direct mechanistic interpretation, even though proportional hazards impose a strong and well-defined structure on the underlying mechanism space.

\section{Mechanism-Based Generative Structure of AFT Models}
\label{app:aft-mechanism}

In this appendix, we provide a formal mechanism-based characterization of accelerated failure time (AFT) models.
This material supports the interpretation given in Section~\ref{sec:aft} and is included here to separate formal generative arguments from the main conceptual discussion.
Our goal is to make explicit the structural assumptions that AFT models impose on the conditional mechanism distribution \(\mathbb{P}(\Theta | X)\).

Let \(\Theta\) denote a latent \emph{individual hazard mechanism} taking values in a measurable mechanism space \(\mathcal{M}\), and let \(\mathcal H(\Theta)=h_\Theta\) denote the induced individual hazard function.
Let \(X\) denote observable covariates.
We assume that there exists a \emph{reference hazard mechanism} with hazard function \(h_0\) and associated survival function
\[
    S_0(t)=\exp\!\left\{-\int_0^t h_0(s)\,ds\right\}.
\]
The defining assumption of AFT models, in the present framework, is that all individual hazard mechanisms are generated from this reference mechanism through time–scale transformations.

\begin{proposition}[Mechanism-based generative structure of AFT models]
\label{prop:aft-mechanism-conditional}
    Assume that, conditional on \(X=x\), the individual hazard mechanism \(\Theta\) is generated through a pure time–scaling transformation driven by a positive random variable \(U>0\), namely
    \[
        \Theta | X=x \;\stackrel{d}{=}\; \Phi_x(U),
        \qquad U \sim P_U,
    \]
    where the mapping \(\Phi_x\) satisfies
    \[
        \mathcal H(\Phi_x(u))(t)
        =
        \frac{1}{u\,a(x)}\,
        h_0\!\left(\frac{t}{u\,a(x)}\right),
        \qquad t \ge 0,
    \]
    and \(a(x)>0\) is a covariate–dependent time–scaling function.
    
    Under this mechanism structure, the individual survival time conditional on \(X=x\) admits the representation
    \[
        T | X=x
        =
        a(x)\,U \cdot T_0,
        \qquad
        T_0 \sim S_0.
    \]
    Defining \(\varepsilon := \log T_0\), we obtain
    \[
        \log T
        =
        \log a(X) + \log U + \varepsilon.
    \]
    
    In this formulation, the conditional mechanism distribution \(\mathbb{P}(\Theta | X)\) is entirely determined by the reference hazard \(h_0\), the time–scaling map \(\Phi_x\), and the distribution of \(U\).
    Consequently, the distribution of the error term \(\varepsilon\) is induced by the reference mechanism through a logarithmic time transformation, rather than being an independent modeling choice.
\end{proposition}

\begin{remark}[Compression of the conditional mechanism distribution]
\label{rem:aft-compression}
    Proposition~\ref{prop:aft-mechanism-conditional} clarifies the fundamental assumptions underlying AFT models at the mechanism level.
    The existence of a single reference hazard mechanism constitutes the primary structural assumption, while the classical error term \(\varepsilon\) arises as a derived quantity.
    
    In particular, all heterogeneity conditional on \(X=x\) is compressed into a one-dimensional time–scaling factor \(U\), so that the conditional mechanism distribution \(\mathbb{P}(\Theta | X=x)\) is confined to the orbit
    \[
        \left\{
            t \mapsto \frac{1}{u}\,h_0\!\left(\frac{t}{u}\right)
            \;:\;
            u>0
        \right\}.
    \]
    As a result, AFT models do not allow for genuine heterogeneity in hazard shape, timing, or modality across individuals.
    All admissible hazard trajectories differ only by acceleration or deceleration along a common reference shape.
    
    This strong compression of the conditional mechanism distribution explains both the interpretability and the limitations of AFT models.
    While the resulting structure facilitates estimation and interpretation, it does not resolve the fundamental non-identifiability of \(\mathbb{P}(\Theta | X)\) discussed in Section~3.
    These two aspects are inseparable consequences of the same structural choice.
\end{remark}

\paragraph{Methodological implications}

The mechanism-based formulation presented above highlights that innovation in survival modeling need not be confined to alternative regression specifications or estimation techniques.
Instead, new model classes may be obtained by relaxing or modifying the structural assumptions imposed on the conditional mechanism distribution \(\mathbb{P}(\Theta | X)\).

From this perspective, generalized AFT models may be constructed by allowing multiple reference mechanisms, multi-dimensional time–scaling factors, or departures from pure time rescaling.
Importantly, such generalizations operate at the level of the mechanism space, rather than through ad hoc modifications of the regression form.
This viewpoint provides a principled framework for balancing interpretability, flexibility, and identifiability in the development of future survival models.

\paragraph{Proof of Proposition \ref{prop:aft-mechanism-conditional}}

\begin{proof}
    Conditioning on \(X=x\), the assumption \(\Theta | X=x \stackrel{d}{=} \Phi_x(U)\) together with the definition of the mapping \(\Phi_x\) implies that, conditional on \(U=u\) and \(X=x\), the individual hazard function takes the form
    \[
        h_{\Theta}(t | U=u, X=x) = \frac{1}{u\,a(x)}\, h_0\!\left(\frac{t}{u\,a(x)}\right).
    \]
    By the deterministic relationship between the hazard and the cumulative hazard, we obtain
    \[
        \int_0^t h_{\Theta}(s | U=u, X=x)\,ds
        =
        \int_0^{t/(u\,a(x))} h_0(v)\,dv
        =
        H_0\!\left(\frac{t}{u\,a(x)}\right),
    \]
    where \(H_0(t)=\int_0^t h_0(s)\,ds\).
    Consequently, the conditional survival function satisfies
    \[
        S(t | U=u, X=x)
        =
        \exp\!\left\{-H_0\!\left(\frac{t}{u\,a(x)}\right)\right\}
        =
        S_0\!\left(\frac{t}{u\,a(x)}\right).
    \]
    Let \(T_0\) be a random variable with survival function \(S_0\), that is, \(\mathbb P(T_0>t)=S_0(t)\).
    For any \(t\ge 0\),
    \[
        \mathbb P\bigl(a(x)\,u\,T_0 > t\bigr)
        =
        \mathbb P\!\left(T_0>\frac{t}{a(x)\,u}\right)
        =
        S_0\!\left(\frac{t}{a(x)\,u}\right),
    \]
    which coincides with the conditional survival function \(S(t | U=u, X=x)\) derived in the previous step.
    Therefore, in distribution,
    \[
        T | (U=u, X=x) \;\stackrel{d}{=}\; a(x)\,u\,T_0,
    \]
    and hence
    \[
        T | X=x = a(x)\,U \cdot T_0, \qquad T_0 \sim S_0,
    \]
    with \(T_0\) independent of \(U\).
    
    Taking logarithms of the above representation and defining \(\varepsilon := \log T_0 \), we obtain
    \[
        \log T = \log a(X) + \log U + \varepsilon.
    \]
    Since \(T_0\) has survival function \(S_0\), for any \(z\in\mathbb R\),
    \[
        \mathbb P(\varepsilon \le z) = \mathbb P(T_0 \le e^z) = 1 - S_0(e^z).
    \]
    Thus, the distribution of \(\varepsilon\) is uniquely determined by \(S_0\) (equivalently, by the reference hazard \(h_0\)) through the logarithmic time transformation.
    
    Combining the above arguments, we conclude that under the proposed mechanism-based generative structure, the conditional mechanism distribution \(\mathbb P(\Theta | X)\) is fully induced by the reference mechanism \(h_0\), the time--scaling map \(\Phi_x\), and the distribution of \(U\).
    Accordingly, the error distribution in the classical AFT representation is not an independent modeling choice, but a direct consequence of the underlying mechanism structure.
\end{proof}

\end{document}